\documentclass[11pt,draftcls,onecolumn]{IEEEtran}

\usepackage{amsfonts}
\usepackage{amssymb}
\usepackage{subfigure}
\usepackage{amsmath,dsfont}
\usepackage{amsthm}
\usepackage{mathrsfs}
\usepackage{verbatim}
\usepackage{graphicx}
\usepackage{epstopdf}
\usepackage{mathtools}
\usepackage{enumerate}
\usepackage{caption}
\usepackage{algpseudocode}
\usepackage{algorithm}
\usepackage{ifthen}
\usepackage{setspace}

\DeclareMathAlphabet{\mathpzc}{OT1}{pzc}{m}{it}

\theoremstyle{remark} \newtheorem{remark}{Remark}

\newcommand{\thmlabel}[1]{\label{thm:#1}}

\theoremstyle{remark} \newtheorem{theorem}{Theorem}

\newcommand{\Vvec}{\mathbf{V}}
\newcommand{\Xvec}{\mathbf{X}}
\newcommand{\Yvec}{\mathbf{Y}}

\newcommand{\tYvec}{\tilde{\mathbf{Y}}}

\newcommand{\mA}{\mathcal{A}}

\newcommand{\mD}{\mathcal{D}}
\newcommand{\mM}{\mathcal{M}}

\newcommand{\sfQ}{\mathsf{Q}}

\newcommand{\realSet}{\mathfrak{R}}

\newcommand{\eps}{\epsilon}
\newcommand{\heps}{\hat{\epsilon}}

\newcommand{\vphi}{\varphi}
\newcommand{\vsig}{\varsigma}

\newcommand{\E}{\mathds{E}}
\newcommand{\Ec}[1]{\mathds{E} \left\{ #1 \right\}}

\newcommand{\mN}{\mathcal{N}}

\newcommand{\sgn}{\text{sgn}}
\newcommand{\dst}{\displaystyle}

\newboolean{SingleColumn}
\setboolean{SingleColumn}{true}

	\ifthenelse{\boolean{SingleColumn}}{}{}

\newcounter{MYtempeqncnt}

\begin{document}

\IEEEoverridecommandlockouts

	\ifthenelse{\boolean{SingleColumn}}{}{}
	\ifthenelse{\boolean{SingleColumn}}{}{}

\author{
\IEEEauthorblockN{Yonathan Murin$^1$, Yonatan Kaspi$^2$, and Ron Dabora$^1$} \\
\authorblockA{\small $^1$Ben-Gurion University, Israel, $^2$University of California, San Diego, USA}

\vspace{-0.90cm}

}




\maketitle
\thispagestyle{plain}
\pagestyle{plain}

\begin{abstract}
	In this work, we consider linear-feedback schemes for the two-user Gaussian broadcast channel with noiseless feedback. We extend the transmission scheme of [Ozarow and Leung, 1984] by applying estimators with memory instead of the memoryless estimators used by Ozarow and Leung (OL) in their original work.
	A recursive formulation of the mean square errors achieved by the proposed estimators is provided, along with a proof for the existence of a fixed point. This enables characterizing the achievable rates of the extended scheme. 
	Finally, via numerical simulations it is shown that the extended scheme can improve upon the original OL scheme in terms of achievable rates, as well as achieve a low probability of error after a finite number of channel uses.

\end{abstract}

\vspace{-0.6cm}
\section{Introduction}

%
%
%
%

\vspace{-0.15cm}
We study the transmission of two independent messages over a two-user Gaussian broadcast channel (GBC) 
with noiseless causal feedback, referred to in the following as the GBCF, focusing on linear-feedback schemes.  
In \cite{OzarowBC:84}, Ozarow and Leung presented inner and outer bounds on the capacity region of the two-user GBCF, and showed that in some scenarios it is larger than the capacity region of the GBC. In the following, we refer to the achievability scheme presented in [1] as the {\em OL scheme}.
%
The OL scheme is a linear-feedback scheme that builds upon the scheme of Schalkwijk and Kailath (SK) \cite{SK:66}, which achieves the capacity of point-to-point (PtP) Gaussian channels with noiseless causal feedback (NCF). 
Motivated by the optimality of the SK scheme for PtP channels, the works \cite{OzarowMAC:84} and \cite{OzarowBC:84} extended this approach to the two-user Gaussian multiple-access channel with NCF (GMACF) and to the two-user GBCF, respectively.
For the GMACF this extension achieves the capacity region, however, for the GBCF this extension is generally suboptimal.
The OL scheme of \cite{OzarowBC:84} and the scheme of \cite{OzarowMAC:84} were later extended to GBCFs and to GMACFs with more than two users as well as to Gaussian interference channels with NCF (GICFs) in \cite{Kramer:02}. 
These schemes were also used in \cite{Zaidi:10} to stabilize (in the mean square sense) two linear, discrete-time, scalar and time-invariant systems in closed-loop, via control over GMACFs and GBCFs, respectively. 

Transmission over the GBCF was also studied using tools from control theory. The work \cite{Elia:04} derived a linear code for the two-user GBCF in which the noises at the receivers are independent. This code obtained higher achievable rates compared to the OL scheme. Later, \cite{AMM:12} used linear quadratic Gaussian (LQG) control to remove the restriction of independent noises in \cite{Elia:04}, and obtained a linear-feedback communications scheme for the $K$-user GBCF.
 Recently, \cite{ASW:14} showed that for the two-user GBCF with independent noises having the same variance, the scheme of \cite{AMM:12} is the optimal scheme subject to using a linear feedback, in the sense of maximal sum-rate.

The work \cite{GLSM:13} studied the GBCF and the GICF and derived a scheme whose achievable sum-rate approaches the full-cooperation bound as the signal-to-noise ratio (SNR) increases to infinity. 
Finally, in \cite{WMW:13} it was shown that the capacity region of the GBCF with independent noises and with only a common message cannot be achieved using a coding scheme which employs linear feedback. 

Note that all the works on GBCFs reviewed above focused on the achievable rates, namely, the rates are obtained as the blocklength increases to infinity. 
From this perspective, it was shown in \cite{AMM:12} that when the noises are independent, the LQG scheme of \cite{AMM:12} achieves a larger rate region than the OL scheme. 
However, in \cite{MKDG:14} we showed that when constraining the blocklength to be finite, the OL scheme can achieve lower mean squared errors (MSEs), and therefore a lower probability of error compared to the LQG scheme (we note that although the focus of \cite{MKDG:14} is on transmission of correlated sources, this observation holds also for independent messages).
In this work we propose an extension of the OL scheme which improves upon the achievable region obtained in \cite{OzarowBC:84}, and benefits from the good performance of the OL scheme when the blocklength is finite.
Next, we detail our main contributions:

{\bf {\slshape Main Contributions}:} 
We focus on linear-feedback schemes as such schemes are simple to implement. 
In the OL scheme of \cite{OzarowBC:84} the receivers' errors are estimated based {\em only} on the last channel output. However, as the transmitted signal in the OL scheme is statistically correlated with {\em all} previous channel outputs, this approach is generally suboptimal.
In this work we provide an {\em explicit} recursive formulation of the minimum MSE (MMSE) estimators which use the last {\em two} channel outputs, along with an explicit recursive characterization of the resulting achievable MSEs. We show that the proposed scheme has a fixed point, which enables characterizing its achievable rates as well as its MSE performance at any finite number of channel uses. 
We note that this is the first {\em explicit characterization} of an OL-type scheme which uses estimators with memory, and the first time that a fixed point property is proved for such a scheme. 
Furthermore, via numerical simulations we show that the extended scheme can both improve upon the original OL scheme in terms of achievable rates, and outperform the scheme of \cite{AMM:12} in terms of probability of error after a {\em finite} number of channel uses.
Finally, we demonstrate that in contrast to the common intuition, applying MMSE estimation based on several recent channel outputs may sometimes result in lower achievable rates than MMSE estimation based only on the most recent channel output.


The rest of this paper is organized as follows: The problem definition and the OL scheme are presented in Section \ref{sec:ProbForm}, the extended OL scheme is derived in Section \ref{sec:OL_extended}, and discussion along with numerical examples are given in Section \ref{sec:examples}.

{\bf {\slshape Notations}:} We use upper-case letters to denote random variables, e.g., $X$, boldface letters to denote random column vectors, e.g., $\Xvec$, and calligraphic letters to denote sets, e.g., $\mM$. 
We use $\Ec{\cdot},(\cdot)^T$ and $\realSet$ to denote the expectation, transpose, and the set of real numbers, respectively.
Lastly, $\sgn(x)$ denotes the sign of $x$, with $\sgn(0) \triangleq 1$.

\vspace{-0.3cm}
\section{Problem Definition and Previous Reslts} \label{sec:ProbForm}

\subsection{Problem Definition} \label{subsec:ProbForm}
\vspace{-0.1cm}
We consider communications over the GBCF, depicted in Fig. \ref{fig:GBC}.
All signals are real. The encoder obtains a pair of independent messages $M_1 \in \mM_1$ and $M_2 \in \mM_2$, where each message is uniformly distributed over its message set. The encoder is required to send the message $M_i, i=1,2$, to the $i$'th receiver, Rx$_i$, using $n$ channel uses. The channel outputs at each receiver at time $k, k=1,2,\dots,n$, are given by:
\vspace{-0.2cm}
\begin{align}
	Y_{i,k} & = X_k + Z_{i,k}, \quad i=1,2, \label{eq:signalModel}
	%
\end{align}

\vspace{-0.15cm}
\noindent where the noises $Z_{i,k} \sim \mN(0, \sigma_{i}^2)$ are i.i.d over time $k$, and independent of $(M_1,M_1)$. Let $\E \{Z_1 Z_2 \} = \rho_z \sigma_{1} \sigma_{2}$.

{\em A $(\mathrm{R}_1, \mathrm{R}_2, n)$ code} for the GBCF consists of
\begin{enumerate}
	\item 
		Two discrete message sets $\mM_i \mspace{-4mu} \triangleq \mspace{-4mu} \{1, 2, \dots, 2^{n\mathrm{R}_i}\}, i=1,2$.
	\item
		An encoder which maps the observed message pair, $M_i \in \mM_i$, and the received NCF up to time $k$, into a channel input at time $k$ via $X_k = f_k(M_1,M_2,\Yvec_{1,1}^{k-1},\Yvec_{2,1}^{k-1})$. 
		%
	\item
		Two decoders $g_{i}: \realSet^{n} \mapsto \mM_i$, each uses its $n$ channel outputs ,$\Yvec_{i,1}^n$, to estimate $M_i$: $\hat{M}_{i} = g_{i}(\Yvec_{i,1}^n)$.
\end{enumerate}

\begin{figure}[t]
    \centering
    \includegraphics[width=0.75\linewidth]{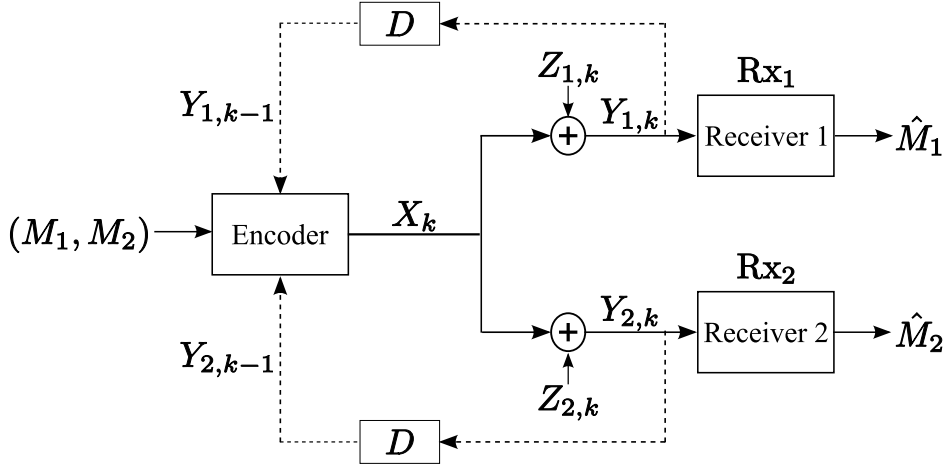}
		\captionsetup{font=footnotesize}
    \caption{The Gaussian broadcast channel with feedback links. The blocks denoted by $D$ represent a unit delay.}
    \label{fig:GBC}
    \vspace{-0.55cm}
\end{figure}

\noindent 
The transmitted signal is subject to the average power constraint \cite{OzarowBC:84}, \cite{AMM:12}:
\vspace{-0.15cm}
\begin{equation}
  \sum_{k=1}^{n} \Ec { X_k^2 } \le nP.
	\label{eq:avgPowerConstraint}
\end{equation}

%
%
\vspace{-0.15cm}
	\noindent The {\em probability of error} at Rx$_i$ is defined as: $P_{e,i}^{(n)} \mspace{-4mu} \triangleq \mspace{-4mu} \Pr \{ \hat{M}_{i} \mspace{-4mu} \neq \mspace{-4mu} M_i \}$. We say that $(\mathrm{R}_1, \mathrm{R}_2)$ is an {\em achievable rate pair} subject to the power constraint \eqref{eq:avgPowerConstraint} if there exists a sequence of $(\mathrm{R}_1, \mathrm{R}_2, n)$ codes satisfying \eqref{eq:avgPowerConstraint}, such that ${\dst \lim_{n \to \infty} P_{e,i}^{(n)} \mspace{-4mu} = \mspace{-4mu} 0}$. 
\noindent Next, we briefly review the OL scheme of~\cite{OzarowBC:84}.


\vspace{-0.3cm}
\subsection{A Short Review of the OL Scheme} \label{subsec:OL_basic}

	\ifthenelse{\boolean{SingleColumn}}{}{}

\vspace{-0.1cm}
In the OL scheme \cite{OzarowBC:84}, prior to transmitting a channel symbol, the transmitter determines the estimation errors at the receivers based on the noiseless feedback, and then sends a linear combination of these errors.
Thus, the channel output at each receiver consists of its estimation error corrupted by a correlated noise term, which consists of the other receiver's error and additive Gaussian noise. Each receiver then updates its estimation according to its observed channel output, thereby, decreasing the variance of its estimation error. 

{\bf {\slshape Setup}:}
Each message $m_i \in \mM_i$ is mapped into a PAM constellation point, $\theta_i$, uniformly distributed over the interval $[-0.5, 0.5]$.
Next, define $\hat{\Theta}_{i,k}$ to be the estimate of the constellation point $\Theta_i$ at the $i$'th receiver, after observing the $k$'th channel output $Y_{i,k}$. Let $\eps_{i,k} \mspace{-3mu} \triangleq \mspace{-3mu} \hat{\Theta}_{i,k} \mspace{-3mu} - \mspace{-3mu} \Theta_i$ be the estimation error after $k$ transmissions, and define $\heps_{i,k-1} \mspace{-3mu} \triangleq \mspace{-3mu} \hat{\Theta}_{i,k-1} \mspace{-3mu} - \mspace{-3mu} \hat{\Theta}_{i,k}$. Thus, we can write $\eps_{i,k} \mspace{-3mu} = \mspace{-3mu} \eps_{i,k-1} \mspace{-3mu} - \mspace{-3mu} \heps_{i,k-1}$. 
 We also define $\alpha_{i,k} \mspace{-3mu} \triangleq \mspace{-3mu} \E \{ \eps_{i,k}^2 \}$ to be the MSEs after $k$ transmissions, and $\rho_k \mspace{-3mu} \triangleq \mspace{-3mu} \frac{\Ec{ \eps_{1,k} \eps_{2,k} }}{\sqrt{\alpha_{1,k}\alpha_{2,k}}}$ to be the correlation coefficient between the estimation errors.

{\bf {\slshape Initialization}:} In the first two transmissions $X_k = \sqrt{12P} \mspace{-4mu} \cdot \mspace{-4mu} \Theta_k, k \mspace{-4mu} = \mspace{-4mu} 1,2$, are sent.
After the first transmission, Rx$_1$ estimates $\Theta_1$ via $\hat{\Theta}_{1,1} \mspace{-4mu} = \mspace{-4mu} \frac{Y_{1,1}}{\sqrt{12P}}$. Rx$_1$ ignores the second transmission and sets $\hat{\Theta}_{1,2} \mspace{-5mu} = \mspace{-5mu} \hat{\Theta}_{1,1}$. Similarly, Rx$_2$ ignores the first transmission and sets $\hat{\Theta}_{2,2} \mspace{-4mu} = \mspace{-4mu} \frac{Y_{2,2}}{\sqrt{12P}}$. Therefore $\alpha_{i,2} \mspace{-4mu} = \mspace{-4mu} \frac{\sigma_i^2}{12P}$, and $\rho_2 \mspace{-3mu} = \mspace{-3mu} 0$.

{\bf {\slshape Encoding}:}
Let $g \mspace{-3mu} > \mspace{-3mu} 0$ be a constant which facilitates a tradeoff between $\mathrm{R}_1$ and $\mathrm{R}_2$, and let $\Psi_{k} \mspace{-4mu} \triangleq \mspace{-4mu} \sqrt{\frac{P}{1 \mspace{-2mu} + \mspace{-2mu} g^2 \mspace{-2mu} + \mspace{-2mu} 2g|\rho_k|}}$.
At the $k$'th transmission, $k \ge 3$, the transmitter sends \cite[pg. 668]{OzarowBC:84}:
\vspace{-0.2cm}
\begin{align}
	X_k \mspace{-3mu} = \mspace{-3mu} \Psi_{k-1} \left( \frac{\eps_{1,k-1}}{\sqrt{\alpha_{1,k-1}}} + \frac{\eps_{2,k-1}}{\sqrt{\alpha_{2,k-1}}} \cdot g \cdot \sgn(\rho_{k-1}) \right),
	\label{eq:OL_txSignal}
\end{align}

\vspace{-0.15cm}
\noindent and the corresponding channel outputs are given by \eqref{eq:signalModel}.


{\bf {\slshape Decoding}:}
Rx$_i$ estimates $\eps_{i,{k-1}}, i=1,2$, based {\em only} on $Y_{i,k}$: $\hat{\eps}_{i,{k-1}} \mspace{-5mu} = \mspace{-5mu} \E \{\eps_{i,{k-1}} | Y_{i,k}\} \mspace{-4mu} = \mspace{-4mu} \frac{\Ec{ \eps_{i,k-1} Y_{i,k} }}{\Ec{ Y_{i,k}^2 }} Y_{i,k}$.
Let $\pi_i \mspace{-3mu} \triangleq \mspace{-3mu} P + \sigma_{i}^2, \Sigma \mspace{-3mu} \triangleq \mspace{-3mu} P + \sigma_{1}^2 + \sigma_{2}^2 - \rho_z\sigma_{1}\sigma_{2}$, and $\vsig_i^2 \mspace{-3mu} \triangleq \mspace{-3mu} \sigma_{i}^2 - \rho_z \sigma_{1} \sigma_{2}$. 
Then, $\alpha_{i,k}$ are given by the recursive expressions \cite[Eqs. (5)--(6)]{OzarowBC:84}:
\vspace{-0.2cm}
\begin{align}
	\alpha_{i,k} & = \alpha_{i,k-1} \frac{\sigma_{i}^2 + \Psi_{k-1}^2 g^{4-2i} (1 - \rho_{k-1}^2)}{\pi_i}, \quad i=1,2, \label{eq:OzarowVar}
	%
\end{align}

\vspace{-0.15cm}
\noindent where the recursive expression for $\rho_k$ is given by \cite[Eq. (7)]{OzarowBC:84}:
\begin{align}
		\rho_k \mspace{-4mu} = \mspace{-4mu} \frac{(\rho_z \sigma_{1} \sigma_{2} \Sigma \mspace{-4mu} + \mspace{-4mu} \vsig_1^2 \vsig_2^2)\rho_{k-1} \mspace{-4mu} - \mspace{-4mu} \Psi_{k-1}^2 \Sigma \cdot g (1 \mspace{-4mu} - \mspace{-4mu} \rho_{k-1}^2) \sgn(\rho_{k-1})}{\sqrt{\pi_1 \pi_2} \sqrt{ \sigma_{1}^2 \mspace{-4mu} + \mspace{-4mu} \Psi_{k-1}^2 g^2 (1 \mspace{-4mu} - \mspace{-4mu} \rho_{k-1}^2)} \sqrt{\sigma_{2}^2 \mspace{-4mu} + \mspace{-4mu} \Psi_{k-1}^2 (1 \mspace{-4mu} - \mspace{-4mu} \rho_{k-1}^2)}}.
		\label{eq:OLRho}
\end{align}

\vspace{-0.1cm}
\noindent In \cite{OzarowBC:84} it was shown that there exists a ${\rho} \in [0,1]$ such that when $|\rho_{k-1}| \mspace{-3mu} = \mspace{-3mu} {\rho}$ then $\rho_k \mspace{-3mu} = \mspace{-3mu} -\rho_{k-1}$. This $\rho$ is a root of the polynomial obtained by setting $\rho_k \mspace{-3mu} = \mspace{-3mu}\rho$ and $\rho_{k-1} \mspace{-3mu} = \mspace{-3mu} -\rho$ in \eqref{eq:OLRho}.
Let $\tilde{\rho}$ denote the largest root of this polynomial in $[0,1]$. 
In \cite{OzarowBC:84} it is shown how to initialize the transmission to guarantee $|\rho_k| \mspace{-3mu} = \mspace{-3mu} \tilde{\rho} \mspace{-3mu} \equiv \mspace{-3mu} \rho_{\text{OL}}, k \mspace{-3mu} \ge \mspace{-3mu} 3$. After $n$ channel uses Rx$_i$ employs a maximum likelihood decoder to recover $M_i$.
Let $\tilde{\Psi} \mspace{-3mu} \triangleq \mspace{-3mu} \frac{P}{1 \mspace{-2mu} + \mspace{-2mu} g^2 \mspace{-2mu} + \mspace{-2mu} 2g\tilde{\rho}}$. Then, the rates achieved by the OL scheme are given by \cite[Eq. (9)]{OzarowBC:84}:
\vspace{-0.15cm}
\begin{align}
	\mathrm{R}_i & \mspace{-3mu} < \mspace{-3mu} \frac{1}{2} \log_2 \mspace{-2mu} \left( \mspace{-2mu} \frac{\pi_i}{\sigma_{i}^2 \mspace{-3mu} + \mspace{-3mu} \tilde{\Psi}^2 g^{4-2i} (1 \mspace{-3mu} - \mspace{-3mu} \tilde{\rho}^2)} \mspace{-2mu} \right) \mspace{-3mu} \triangleq \mspace{-3mu} \mathrm{R}_i^{\text{OL}}.  \label{eq:OLrates}
\end{align}

\vspace{-0.35cm}
\section{A New Extended OL Scheme} \label{sec:OL_extended}

\vspace{-0.1cm}
The MMSE estimator of $\eps_{i,k-1}$ based on the channel outputs $\Yvec_{i,1}^k$, is given by $\E \{\eps_{i,k-1} | \Yvec_{i,1}^k \}$.
Yet, as successive channel outputs are not independent, obtaining an explicit expression for this estimator is analytically intractable. 
In the OL scheme, the estimates of $\eps_{i,k-1}$ are generated based {\em only} on $Y_{i,k}$. 
These estimators are suboptimal since $\Yvec_{i,1}^{k-1}$ and $\eps_{i,k-1}$ are correlated.
A natural way to improve upon the OL scheme is estimating $\eps_{i,k-1}$ based on $[Y_{i,k}, Y_{i,k-1}]^T \mspace{-3mu} \triangleq \mspace{-3mu} \tYvec_{i,k}$. We refer to this as extended OL (EOL). The EOL encoding is done as in \eqref{eq:OL_txSignal}.
Let $\sfQ_{\tYvec_{i,k}}$ denote the covariance matrix of the vector $\tYvec_{i,k}$.
	%
	\noindent Since $ \eps_{i,k-1}$ and $\tYvec_{i,k}$ are jointly Gaussian,
	the MMSE estimator of $\epsilon_{i,k-1}$ based on $\tYvec_{i,k}$ is given~by \cite[Eq. (12.6)]{Kay}:
	\vspace{-0.2cm}
	\begin{equation}
		\hat{\eps}_{i,{k-1}} = \Ec{ \eps_{i,k-1} \cdot (\tYvec_{i,k})^T } \cdot \sfQ_{\tYvec_{i,k}}^{-1} \cdot \tYvec_{i,k}.
	\label{eq:LmmseEstWithMem}
	\end{equation}
	
	\vspace{-0.15cm}
	\noindent The following theorem explicitly characterizes $\hat{\eps}_{i,{k-1}}$ in \eqref{eq:LmmseEstWithMem}:
	\begin{theorem}
		\thmlabel{thm:ImprovedOL}
			\ifthenelse{\boolean{SingleColumn}}{}{}
			\ifthenelse{\boolean{SingleColumn}}{}{}
		
	
	\end{theorem}
	
	\begin{proof}[$\mspace{-20mu}$ Proof outline]
		Let $\lambda_{i,k-1}$ denote the off-diagonal elements of $\sfQ_{\tYvec_{i,k}}$ (the two off-diagonal elements of $\sfQ_{\tYvec_{i,k}}$ are equal). Explicit direct calculation of $\hat{\eps}_{i,{k-1}}$, in terms of $\rho_{k-1}, \alpha_{i,k-1}$ and $\lambda_{i,k-1}$, results in \eqref{eq:heps_dualMem}. The recursive expressions in \eqref{eq:lambdas_def} are then obtained via an explicit calculation of $\E \{Y_{i,k} Y_{i,k-1} \}$, and the instantaneous MSEs in \eqref{eq:alpha_duoMem} are calculated via $\E \{ \eps_{i,k}^2 \}$.		
		Finally, the instantaneous correlation coefficient is calculated via $\rho_k \mspace{-3mu} \triangleq \mspace{-3mu} \frac{\Ec{ \eps_{1,k} \eps_{2,k} }}{\sqrt{\alpha_{1,k}\alpha_{2,k}}}$.
	\end{proof}
	
	\begin{remark}
		Fixing $\lambda_{i,k} \mspace{-1mu} = \mspace{-1mu} 0, k \ge 1$, EOL specializes to OL.
	\end{remark}
	
	Similarly to the OL scheme, the EOL scheme has a fixed-point, which is stated in the following theorem:
	\begin{theorem}
    \thmlabel{thm:fixedPoinfBothMem}
		Consider the EOL scheme with the decoders given in \eqref{eq:heps_dualMem}--\eqref{eq:rho_terms_ext} and encoding given in \eqref{eq:OL_txSignal}. Then, there exists a $(\rho, \lambda_1,\lambda_2) \in [0,1]\times \realSet^2$ such that if $|\rho_{k-1}| = {\rho}, \lambda_{i,k-1}=\lambda_i, i=1,2$, then $|\rho_k| = \rho, \lambda_{i,k}=\lambda_i,i=1,2$.
	\end{theorem}
	
	%

	\begin{proof}[$\mspace{-20mu}$ Proof]
		First, note that the method used to prove the fixed point for the OL scheme cannot be applied to the EOL due to the terms $\lambda_{i,k-1}$, cf. \cite[pg. 669]{OzarowBC:84}. 		
		
		The fixed point is proven by applying Brouwer's fixed-point theorem, \cite[Subsection 12.8.4]{MathHandbook}, to the estimation scheme \eqref{eq:lambdas_def}--\eqref{eq:rho_terms_ext}. 
		Let $\xi_k \mspace{-3mu} \triangleq \mspace{-3mu} \sgn(\rho_{k}) \sgn(\rho_{k-1}) \mspace{-3mu} \in \mspace{-3mu} \{1,-1 \}$, and define the vector $\Vvec_k \triangleq [\lambda_{1,k}, \lambda_{2,k}, \rho_k^2, \xi_k]$.  
		Eqs. \eqref{eq:heps_dualMem}--\eqref{eq:rho_terms_ext} imply that $\Vvec_{k-1}$ determines $\lambda_{i,k}$ and $\rho_k$. 
		Let $\nu$ denote the mapping from $\Vvec_{k-1}$ to $\Vvec_{k}$ and let $\nu_1$ denote the mapping from $\Vvec_{k-1}$ to $\Vvec_{k}$ when $\xi_k \mspace{-4mu} = \mspace{-4mu} 1, \forall k$. 
		We prove that $\nu$ has a fixed point in two steps: First, we show that $\nu_1$ has a fixed point. Then, we show that a fixed point of $\nu_1$ translates into a fixed point of~$\nu$.
		
		{\bf {\slshape Fixed point of $\nu_1$}:}
		Assume that $\xi_k=1, \forall k$, and define $\Vvec_{1,k} \triangleq [\Vvec_k]_{\xi_k=1}$. 
		We show that $\Vvec_{1,k} = \nu_1(\Vvec_{1,k-1})$, i.e., knowledge of $\Vvec_{1,k-1}$ and constants is sufficient to calculate $\Vvec_{1,k}$. 
		Eq. \eqref{eq:Omega_k1_def} implies that $\Omega_{k-1}^2$ is a function of $\rho_{k-1}^2, \lambda_{1,k-1}^2$ and $\lambda_{2,k-1}^2$. Similarly, from \eqref{eq:Theta_k1_def} we have that $T_{k-1}^2$ is also a function of $\rho_{k-1}^2, \lambda_{1,k-1}^2$ and $\lambda_{2,k-1}^2$. 
		%
	%
		\noindent Therefore, for $\xi_{k-1}=1$ we have that $\rho_k^2$ can be obtained from $\rho_{k-1}^2, \lambda_{1,k-1}$ and $\lambda_{2,k-1}$. From \eqref{eq:lambdas_def} it follows that for $\xi_{k-1}=1$, $\lambda_{i,k}$ are functions of $\rho_{k-1}^2, \lambda_{1,k-1}$ and $\lambda_{2,k-1}$. 
			
		Noting that $\lambda_{i,k}^2 < \pi_i^2, \forall k$, we conclude that for $\mA \triangleq [-\pi_1,\pi_1] \times [-\pi_1,\pi_2] \times [0,1]$ the mapping $\nu_1$ obeys $\nu_1: \mA \mapsto \mA$. Finally, recall Brouwer's fixed-point theorem which states that if $\mD$ is convex and compact and $h: \mD \mapsto \mD$ is a continuous function, then $h$ has a fixed point. As $\mA$ is compact and convex, it follows that $\nu_1$ has a fixed point. We denote this fixed point by~$\bar{\Vvec}_1 = [\bar{\lambda_1}, \bar{\lambda_2}, \bar{\rho}^2]$.
		
		{\bf {\slshape Fixed point of $\nu$}:} 
		We show that $\nu([\bar{\Vvec}_1, 1]) \mspace{-3mu} = \mspace{-3mu} [\bar{\Vvec}_1, 1]$. 
		As $\bar{\Vvec}_1$ is a fixed point of $\nu_1$, it follows that if $\rho^2_{k-1} \mspace{-3mu} = \mspace{-3mu} \bar{\rho}^2, \lambda_{i,k-1} \mspace{-3mu} = \mspace{-3mu} \bar{\lambda}_{i}$, then $\rho^2_{k} \mspace{-3mu} = \mspace{-3mu} \bar{\rho}^2, \lambda_{i,k} \mspace{-3mu} = \mspace{-3mu} \bar{\lambda}_{i}$. 
		Therefore, as $\lambda_{1,k}=\lambda_{1,k-1}$, \eqref{eq:lambda1_km1_recursive} implies that if $\xi_{k-1} \mspace{-4mu} = \mspace{-4mu} \bar{\xi}$ then $\xi_{k} \mspace{-4mu} = \mspace{-4mu} \bar{\xi}$. Thus, $\nu$ has a fixed point. 
		The proof is the same for $\xi_k \mspace{-4mu} = \mspace{-4mu} -1$.				
	\end{proof}
	
	%
	
	Let $\bar{V}=[\bar{\lambda}_1, \bar{\lambda}_2, \bar{\rho}^2, \bar{\xi}]$ be a fixed point of $\nu$, and let $\bar{\Psi} \mspace{-3mu} \triangleq \mspace{-3mu} \frac{P}{1 \mspace{-2mu} + \mspace{-2mu} g^2 \mspace{-2mu} + \mspace{-2mu} 2g\bar{\rho}}$. 
	Similarly to \cite[pg. 669]{OzarowBC:84} the initialization procedure can be designed to guarantee $|\rho_2| \mspace{-3mu} = \mspace{-3mu} \bar{\rho} \mspace{-3mu} \equiv \mspace{-3mu} \rho_{\text{EOL}}$. Further setting $\lambda_{i,2}=\bar{\lambda}_i$ will result in $|\rho_k|=\bar{\rho}$ and $\lambda_{i,k}=\bar{\lambda}_i$ for $k\ge 3$. Therefore, the EOL scheme achieves rate pairs satisfying:
	\vspace{-0.15cm}
	\begin{align}
		\mspace{-10mu} \mathrm{R}_i & \mspace{-4mu} < \mspace{-4mu} \frac{1}{2} \mspace{-2mu} \log \mspace{-3mu} \left( \mspace{-3mu}  \frac{\pi_i^2 - \bar{\lambda}_{i}^2}{\pi_i^2 \mspace{-4mu} - \mspace{-4mu} \bar{\lambda}_{i}^2 \mspace{-3mu} - \mspace{-3mu} P \pi_i \mspace{-3mu} + \mspace{-3mu} \bar{\Psi}^2 g^{4-2i} (1 \mspace{-4mu} - \mspace{-4mu} \bar{\rho}^2 \mspace{-1mu} ) \pi_i} \mspace{-3mu} \right) \mspace{-4mu} \triangleq \mspace{-4mu} \mathrm{R}_i^{\text{EOL}}.  \label{eq:OLrates_ext} 
		%
	\end{align}
	
	
	\vspace{-0.35cm}
	\section{Numerical Examples and A Discussion} \label{sec:examples}

	\vspace{-0.1cm}
	\subsection{The Acheivable Rate Region} \label{subsec:examples_rates}
	\vspace{-0.05cm}
	Consider the GBCF with $\sigma_1^2 \mspace{-4mu} = \mspace{-4mu} \sigma_2^2 \mspace{-4mu} = \mspace{-4mu} 1, \rho_z \mspace{-4mu} = \mspace{-4mu} 0$, and $P \mspace{-4mu} = \mspace{-4mu} 5$. Fig. \ref{fig:RateRegionSymm} illustrates the achievable rate regions of the OL scheme, the EOL scheme, and the LQG scheme of \cite[Thm. 1]{AMM:12}. The regions for OL and EOL are obtained by varying $g$ in the range $[0.01, 100]$. It can be observed that in this setting EOL outperforms OL, and that LQG outperforms both OL and EOL. 
	The subfigure in Fig. \ref{fig:RateRegionSymm} depicts $\rho_{\text{OL}}$ and $\rho_{\text{EOL}}$ versus $g$, for the same setting. It can be observed that $\rho_{\text{EOL}} \le \rho_{\text{OL}}$. 
	The intuition for this relationship is as follows: since the estimator \eqref{eq:heps1_dualMem} uses $Y_{1,k}$ and $Y_{1,k-1}$ for estimation, and since $Y_{1,k-1}$ is correlated with $\eps_{2,k-1}$, this reduces the correlation between $\eps_{1,k}=\eps_{1,k-1}-\hat{\eps}_{1,k-1}$ and $\eps_{2,k}=\eps_{2,k-1}-\hat{\eps}_{2,k-1}$, which leads to $\rho_{\text{EOL}} \le \rho_{\text{OL}}$. 
	
	\ifthenelse{\boolean{SingleColumn}}{}{}
	\ifthenelse{\boolean{SingleColumn}}{}{}
	
	
	Next, note that {\em in some scenarios OL can outperform EOL}. 
	The reason for this situation is that the achievable rates in the OL and EOL schemes are subject to two contradicting effects: while the subtraction of $\bar{\lambda}_{i}^2$ in the numerator and denominator of (12) increases $\mathrm{R}_i^{\text{EOL}}$ compared to $\mathrm{R}_i^{\text{OL}}$ (which corresponds to $\bar{\lambda}_i^2 = 0$), the fact that $\rho_{EOL}$ can be smaller than $\rho_{OL}$ can decrease $\mathrm{R}_i^{\text{EOL}}$ compared to $\mathrm{R}_i^{\text{OL}}$ (this follows as both $\mathrm{R}_i^{\text{OL}}$ and $\mathrm{R}_i^{\text{EOL}}$ increase with $\rho_{\text{OL}}$ and $\rho_{\text{EOL}}$, respectively).
	This situation is illustrated in Fig. 3 which presents the achievable rate regions for $\sigma_1^2 \mspace{-4mu} = \mspace{-4mu} 0.1, \sigma_2^2 \mspace{-4mu} =\mspace{-4mu} 50,\rho_z \mspace{-4mu} = \mspace{-4mu} 0$ and $P \mspace{-4mu} = \mspace{-4mu} 1$. It can be observed in the figure that for large $\mathrm{R}_1$ and small $\mathrm{R}_2$ OL outperforms EOL. 
	Finally, note that the OL and EOL schemes can be combined by applying a decoder which uses the estimator that achieves the largest $\mathrm{R}_2$ at any specific $\mathrm{R}_1$.

	\ifthenelse{\boolean{SingleColumn}}{}{}
	\ifthenelse{\boolean{SingleColumn}}{}{}
	
	\vspace{-0.35cm}
	\subsection{Probability of Error for Finite Blocklengths}
	
	\ifthenelse{\boolean{SingleColumn}}{}{}
	\ifthenelse{\boolean{SingleColumn}}{}{}
	
	\vspace{-0.05cm}
	Motivated by the results of \cite{MKDG:14}, in this subsection we consider the {\em finite blocklength} regime, which implies $P_{e,i}^{(n)} \mspace{-4mu} > \mspace{-4mu} 0$. 
	
	For independent noises with equal variances, the LQG scheme is a realization of the class of schemes presented in \cite{ASW:14}, which achieves the highest {\em sum-rate} among all linear-feedback schemes.
	Furthermore, for this setting the LQG scheme is also a realization of the class of schemes presented in \cite{Elia:04}.
	In fact, \cite{AMM:12} showed that for this setting, {\em in terms of achievable rates}, LQG strictly outperforms OL, as is demonstrated in Fig. \ref{fig:RateRegionSymm}. 
	Recall that in the OL and in the EOL schemes, the achievable rates are determined by the scheme's steady-state (fixed point) in terms of $\rho_k^2$ (and $\lambda_{i,k}$). In this steady-state, at each channel use the MSE $\alpha_{i,k}$ is attenuated by a {\em constant factor}, which determines the achievable rates, see \cite[Lemma 1]{AMM:12} on the connection between the MSEs and the achievable rates. 
	Similarly, the achievable rates of the LQG scheme are determined by the scheme's steady-state MSE exponents. 
	However, numerical evaluations show that the LQG scheme converges to its steady-state slower than the OL and EOL schemes.
	Based on this observation, \cite{MKDG:14} showed that when the codeword length is finite, the OL scheme can achieve lower MSE compared to the LQG scheme.
	Furthermore, it can be easily observed that if $\mathrm{R}_i^{\text{EOL}} \mspace{-3mu} > \mspace{-3mu} \mathrm{R}_i^{\text{OL}}$, and $\rho_{\text{EOL}} \mspace{-3mu} < \mspace{-3mu} \rho_{\text{OL}}$ (as indicated in Fig. \ref{fig:RateRegionSymm}), then EOL outperforms OL also in the finite blocklength regime.
	
	Let $\beta_{i,n}$ denote the MSE achieved by a decoder of a linear-feedback transmission scheme after $n$ channel uses, and let $\mathrm{R}_i$ be the transmission rate.
	Recall that as the scheme is linear the estimation error is a Gaussian RV \cite[Subsection 10.5]{Kay}.
	Since the data points are selected out of a PAM constellation over $[-0.5,0.5]$, the probability of error can be computed using the standard expression for PAM \cite[pg. 670]{OzarowBC:84}:
	\vspace{-0.3cm}
	\begin{align}
		P_{e,i}^{(n)} = \frac{ 2^{n\mathrm{R}_i} - 1}{2^{n\mathrm{R}_i - 1}} Q \Bigg( \frac{1}{2^{n\mathrm{R}_i + 1} \sqrt{\beta_{i,n}}} \Bigg).
	\end{align}
	
	\vspace{-0.28cm}
	\noindent 
	Let $\mathrm{R}_1^{\text{OL}}(\rho_z)$ denote the achievable rate of the OL scheme at a specific noise correlation $\rho_z$,
	and similarly define $\mathrm{R}_1^{\text{EOL}}(\rho_z)$ and $\mathrm{R}_1^{\text{LQG}}(\rho_z)$. 
	Fig. \ref{fig:Perror_Combine} depicts $P_{e,1}^{(n)}$ vs. $n$ for the OL, EOL and LQG schemes, for $P \mspace{-3mu} = \mspace{-3mu} 2, \sigma_1^2 \mspace{-3mu} = \mspace{-3mu} \sigma_2^2 \mspace{-3mu} = \mspace{-3mu}1$, and $g \mspace{-3mu} = \mspace{-3mu} 1$, for two cases: $\rho_z=0$, and $\rho_z =0.3$. For this setting $\mathrm{R}_1^{\text{OL}}(0) \mspace{-3mu} = \mspace{-3mu} 0.458$, $\mathrm{R}_1^{\text{EOL}}(0) \mspace{-3mu} = \mspace{-3mu} 0.461$, and $\mathrm{R}_1^{\text{LQG}}(0) \mspace{-3mu} = \mspace{-3mu} 0.464$. The transmission rate, for {\em all the schemes}, is set to $\mathrm{R}_1 \mspace{-3mu} = \mspace{-3mu} 0.9 \cdot \mathrm{R}_1^{\text{OL}}(\rho_z), \rho_z \mspace{-3mu} = \mspace{-3mu} 0,0.3$. It can be observed that, for $\rho_z \mspace{-5mu} = \mspace{-5mu} 0$, the EOL scheme achieves $P_{e,1}^{(n)} \mspace{-4mu} = \mspace{-4mu} 10^{-5}$ after $n \mspace{-4mu} = \mspace{-4mu} 18$ channel uses, while the OL and LQG schemes require $n \mspace{-4mu} = \mspace{-4mu} 20$ and $n \mspace{-4mu} = \mspace{-4mu} 56$ channel uses, respectively. 
	It can be further observed that for small $n$ the EOL scheme and the OL scheme achieve similar $P_{e,1}^{(n)}$; however, for larger $n$ the EOL scheme significantly improves upon the OL scheme. 
	These observations also hold when the noises are correlated, as concluded from the curves corresponding to $\rho_z \mspace{-4mu} = \mspace{-4mu} 0.3$ in~Fig.~\ref{fig:Perror_Combine}.
	
	Finally, note that for $\rho_z \mspace{-3mu} = \mspace{-3mu} 0$, a fixed transmission rate $0.9 \mspace{-3mu} \cdot \mspace{-3mu} \mathrm{R}_1^{\text{OL}}(0)$, and $\sigma_1^2 \mspace{-3mu} = \mspace{-3mu} \sigma_2^2 \mspace{-3mu} = \mspace{-3mu}1$, the LQG scheme requires $P \mspace{-3mu} = \mspace{-3mu} 2.8$ in order to achieve $P_{e,1}^{(n)} \mspace{-4mu} = \mspace{-4mu} 10^{-5}$ after $n \mspace{-4mu} = \mspace{-4mu} 18$ channel uses. This reflects an SNR loss of 1.46 dB compared to the EOL scheme. 
	{\em We conclude that in the finite blocklength regime the EOL scheme can significantly improve upon both the OL and the LQG schemes}.
	
	%
	%
		%


\vspace{-0.3cm}		
\begin{thebibliography}{10}


\bibitem{OzarowBC:84}
L. H. Ozarow and S. K. Leung-Yan-Cheong,
\newblock{``An achievable region and outer bound for the Gaussian broadcast channel with feedback,"}
\newblock{\em IEEE Trans. Inf. Theory.}, vol. 30, no. 4, pp. 667--671, Jul. 1984.

\bibitem{SK:66}
J.~P.~M.~Schalkwijk and T.~Kailath,
\newblock{``A coding scheme for additive white noise channels with feedback--Part I: No bandwidth constraint,"}
\newblock{\em IEEE Trans. Inf. Theory.}, vol. 12, no. 2, pp. 172--182, Apr. 1966.



\bibitem{OzarowMAC:84}
L. H. Ozarow,
\newblock{``The capacity of the white Gaussian multiple access channel with feedback,"}
\newblock{\em IEEE Trans. Inf. Theory.}, vol. 30, no. 4, pp. 623--629, Jul. 1984.


\bibitem{Kramer:02}
G. Kramer,
\newblock{``Feedback strategies for white Gaussian interference networks,"}
\newblock{\em IEEE Trans. Inf. Theory.}, vol. 48, no. 6, pp. 1423--1438, Jun. 2002.

\bibitem{Zaidi:10}
A.~A.~Zaidi, T.~J.~Oechtering and M.~Skoglund,
\newblock{``Sufficient conditions for closed-loop control over multiple-access and broadcast channels,"}
\newblock{\em Proc. IEEE Conf. on Decision and Cont.}, Atlanta, GA, Dec. 2010, pp. 4771--4776.


\bibitem{Elia:04}
N.~Elia,
\newblock{``When Bode meets Shannon: Control oriented feedback communication schemes,"}
\newblock{\em IEEE Trans. Automat. Control}, vol. 49, no. 9, pp. 1477--1488, Sep. 2004.

\bibitem{AMM:12}
E.~Ardestanizadeh, P.~Minero, and M.~Franceschetti,
\newblock{``LQG control approach to Gaussian broadcast channels with feedback,"}
\newblock{\em IEEE Trans. Inf. Theory}, vol. 58, no. 8, pp. 5267--5278, Aug. 2012.

\bibitem{ASW:14}
S.~B.~Amor, Y.~Steinberg, and M.~Wigger,
\newblock{``Duality with linear-feedback schemes for the scalar Gaussian MAC and BC,"}
in \newblock{\em Proc. Int. Zurich Seminar Commun.}, Zurich, Switzerland, Feb. 2014, pp. 25--28.

\bibitem{GLSM:13}
M.~Gastpar, A.~Lapidoth, Y.~Steinberg, and M.~Wigger,
\newblock{``Coding schemes and asymptotic capacity for the Gaussian broadcast and interference channels with
feedback,"}
\newblock{\em IEEE Trans. Inf. Theory}, vol. 60, no. 1, pp. 54--71, Jan. 2014.

\bibitem{WMW:13}
Y.~Wu, P.~Minero, and M.~Wigger,
\newblock{``Insufficiency of linear-feedback schemes in Gaussian broadcast channels with common message,"}
\newblock{\em IEEE Trans. Inf. Theory}, vol. 60, no. 8, pp. 4553--4566, Aug. 2014.

\bibitem{MKDG:14}
Y.~Murin, Y.~Kaspi, R.~Dabora and D.~G\"und\"uz,
\newblock{``Uncoded transmission of correlated Gaussian sources over broadcast channels with feedback,"}
in \newblock{\em Proc. IEEE GlobalSIP Symp. on Network Theory}, Atlanta, GA, Dec. 2014, pp. 1063--1067.

\bibitem{Kay}
S.~M.~Kay,
\newblock{\em Fundamentals of Statistical Signal Processing: Estimation Theory}.
\newblock{Englewood Cliffs, NJ: Prentice Hall, 1993.}

\bibitem{MathHandbook}
I.~N.~Bronshtein, K.~A.~Semendyayev, G.~Musiol and H.~Muehlig
\newblock{\em Handbook of Mathematics}.
\newblock{5th ed. Springer, 2007.}


%
%
%
%
%
%
%
%

 
\end{thebibliography}
\end{document}